\documentclass{llncs}

\usepackage[utf8]{inputenc} 

\usepackage{fullpage}

\usepackage{amssymb}

\begin{document}

\title{Local Improvement Gives Better Expanders} \author{Michael Lampis \thanks{Research supported by ERC Grant 226203}}
\institute{KTH Royal Institute of Technology}

\maketitle 

\begin{abstract}

It has long been known that random regular graphs are with high probability
good expanders. This was first established in the 1980s by Bollobás by directly
calculating the probability that a set of vertices has small expansion and then
applying the union bound.

In this paper we improve on this analysis by relying on a simple high-level
observation: if a graph contains a set of vertices with small expansion then it
must also contain such a set of vertices that is locally optimal, that is, a
set whose expansion cannot be made smaller by exchanging a vertex from the set
with one from the set's complement. We show that the probability that a set of
vertices satisfies this additional property is significantly smaller. Thus,
after again applying the union bound, we obtain improved lower bounds on the
expansion of random $\Delta$-regular graphs for $\Delta\ge 4$. In fact, the
gains from this analysis increase as $\Delta$ grows, a fact we explain by
extending our technique to general $\Delta$. Thus, in the end we obtain an
improvement not only for some small special cases but on the general asymptotic
bound on the expansion of $\Delta$-regular graphs given by Bollobás.

\end{abstract}

\newpage

\setcounter{page}{1}

\section{Introduction}

A graph is called an expander graph if for all partitions of its set of
vertices into two sets the ratio of the number of edges with endpoints in both
sets over the size of the smaller set is bounded from below by a constant
independent of the graph size. Roughly speaking, an expander is a very
well-connected graph and we are usually interested in expanders which are
regular and sparse. Expander graphs have long been an object of intense study
in discrete mathematics and theoretical computer science. Over time, they have
been used in a wide number of applications, such as for example
error-correcting codes \cite{sipser1996expander}, computational complexity
\cite{dinur2007pcp,reingold2008undirected}, inapproximability for
bounded-occurrence CSPs \cite{papadimitriou1991optimization} and others too
numerous to mention here (see \cite{hoory2006expander}).

One of the most fundamental results in this area, established by Bollobás in
1988 \cite{bollobas1988isoperimetric}, is that random $\Delta$-regular graphs
are expanders with high probability. Thus, though the existence of some
expander graph families had already been known
\cite{buser1984bipartition,lubotzky1988ramanujan}, Bollobás managed to show
that in fact almost all $\Delta$-regular graphs are expanders using
surprisingly elementary methods (straightforward counting arguments, as opposed
to the linear algebra techniques commonly used when dealing with expander
graphs).  Furthermore, Bollobás gave an asymptotic lower bound of
$\frac{\Delta}{2}-\Theta(\sqrt{\Delta})$ on the expansion of random
$\Delta$-regular graphs, and specific bounds for various small values of
$\Delta$. Interestingly, the asymptotic bound was later found to be optimal up
to the coefficient of $\sqrt{\Delta}$, as Alon showed that any $\Delta$-regular
graph (not just random graphs) must have a bi-section with expansion at most
$\frac{\Delta}{2}-\Theta(\sqrt{\Delta})$ \cite{alon1997edge}.

The main objective of this paper is to revisit this classical result of
Bollobás and to improve it through a more refined analysis. With this new
analysis we will show higher lower bounds on the expansion of random
$\Delta$-regular graphs. 

Our results move in two main directions. First, we will obtain better lower
bounds on the expansion of $\Delta$-regular graphs for various small values of
$\Delta$. One of the reasons we are interested in this is because improvements
in the currently known bounds could have applications in other areas, and in
particular one example is the inapproximability of bounded occurrence CSPs. For
instance, the fact that 6-regular graphs have expansion at least 1 (established
in \cite{bollobas1988isoperimetric}) is invoked a number of times in this area
(see \cite{BK03,BK01,BK99}). It is still open whether $\Delta=6$ is the
smallest degree for which this is true, and resolving this question would be of
immediate use in inapproximability reductions. Of course, beyond this, finding
what is the expansion of random $\Delta$-regular graphs for small values of
$\Delta$ is in itself a fundamental mathematical problem, intimately connected
to a number of other basic graph-theoretic questions such as the size of the
minimum bisection of random regular graphs. Unfortunately, little progress has
been made since \cite{bollobas1988isoperimetric} and in most cases the bounds
given by Bollobás are still the best known lower bounds, although some progress
has been made in establishing corresponding upper bounds through results on the
minimum bisection problem \cite{monien2001upper,diaz2003bounds,diaz2007bounds}.

Second, we will use our refined analysis to also improve the asymptotic bound
on the expansion as a function of $\Delta$, or more precisely, to show that the
current bound is not tight.  As mentioned, it is known that the correct value
is $\frac{\Delta}{2}-\Theta(\sqrt{\Delta})$ but the coefficients of
$\sqrt{\Delta}$ given by Bollobás in the lower bound ($\sqrt{\ln 2} \approx
0.83$) and Alon in the upper bound ($\frac{3}{16\sqrt{2}}\approx 0.13$) are
relatively far apart. Again, this is a fundamental question given as an open
problem in \cite{alon1997edge}. But furthermore, the main reason we would like
to give such a general analysis is to show that the refinements we make in the
analysis of small values of $\Delta$ are more than just an ad-hoc trick that
improves some small special cases, but in fact our technique applies to all
$\Delta$.

To explain the high-level idea of our method let us first recall the proof of
the main result of \cite{bollobas1988isoperimetric}. The idea there is to
calculate the probability that a certain fixed set of vertices of a random
$\Delta$-regular graph has exactly $c$ edges connecting it to its complement.
Clearly, as $c$ becomes smaller that probability also becomes smaller, so the
trick here is to calculate the largest value of $c$ that still makes the
probability $o(2^{-n})$. Afterwards, an application of the union bound ensures
that with high probability no set has $c$ or fewer edges crossing the cut,
giving a bound on the expansion. The best bound on $c$ is found by directly
writing out the probability that exactly $c$ edges cross the cut and performing
a tedious but straightforward calculation.

Since the probability that $c$ edges cross the cut is calculated exactly in
\cite{bollobas1988isoperimetric} one might expect that the best place to look
for an improvement may be in the application of the trivial union bound, which
completely ignores a very basic fact: The expansion ratios of two different
sets of vertices are far from independent. Somewhat surprisingly, we give an
argument that attempts to exploit this crucial fact while still relying on the
union bound.  The main idea is that if a set with small expansion exists, then
there must exist a set with the same or smaller expansion which is locally
optimal. A set is locally optimal in this context if exchanging a vertex from
it with a vertex from its complement cannot decrease its expansion.  Thus, to
establish that a graph's expansion is above a certain value it suffices to show
that no locally optimal set exists with expansion below this value. Again, we
want to find the maximum $c$ such that the probability that a fixed set of
vertices is locally optimal and has $c$ edges leaving it is $o(2^{-n})$ and, as
might be expected, this probability turns out to be significantly smaller than
the probability that the set simply has $c$ edges leaving it, allowing us to
end up with a larger $c$.

The main technical challenge arising here is to calculate this probability
accurately. In the case of \cite{bollobas1988isoperimetric} it is a simple
counting exercise to determine the number of configurations where exactly $c$
edges cross the cut obtaining a clean and relatively simple formula. In our
case however, as we will argue, the local optimality condition imposes some
constraints on the degrees of the vertices which make counting much more
complex. Specifically, in a locally optimal set all vertices have the majority
of their neighbors in their own set. Because of this complication a clean
closed formula for the number of configurations that give a locally optimal set
with $c$ edges coming out probably does not exist.  Given that, our main
technical effort is concentrated on determining the degree distribution of
maximum probability. Using a (different kind of) local optimality argument we
show that the number of neighbors each vertex has on the other side of the
partition can be assumed to follow roughly a binomial distribution, with some
undetermined parameters (Lemma \ref{lem:localmax}).

Having arrived at this technical tool we then apply it to the problem at hand.
First, as a warm-up we show that if one ignores the degree constraints it is
not hard to calculate the parameters of the degree distribution. We thus, after
some calculations, arrive at exactly the same bound given in
\cite{bollobas1988isoperimetric}, albeit through a rockier road and gaining
some valuable insight in the kind of configurations that lead to bad cuts. Then
we go on to exploit this insight combined with the local improvements idea to
improve on the expansion lower bounds for specific values of $\Delta$. This
leads to a set of optimization problems which we solve numerically to obtain
our new results (note that the correctness of the bounds we give can however be
verified by simple calculations).  Our bounds improve on those given in
\cite{bollobas1988isoperimetric} and for $\Delta\ge 4$ improve the state of the
art (for cubic graphs the best result is still given in \cite{kostochka}, but
the complicated analysis there is restricted only to $\Delta=3$).  Finally, we
show how our analysis can be applied for general $\Delta$ and that an improved
asymptotic lower bound on the expansion of random $\Delta$-regular graphs can
be obtained. The new coefficient of $\Delta$ we get offers only a small
improvement over the one given in \cite{bollobas1988isoperimetric}, so we do
not expend much effort trying to calculate it exactly. Instead, for the sake of
simplicity we simply show that a strictly smaller coefficient is achievable.
Despite the lack of a large improvement in this constant, we believe there is
still interest in this result, in that it indicates that the bound given by
Bollobás is not the "right" bound, and also in that it shows that the local
improvement method we employ here applies generally and not just for small
values of $\Delta$.

\section{Definitions and Notation}

We use standard graph terminology and will be dealing with regular graphs only.
We use $\Delta$ to denote the degree, $n$ to denote the number of vertices of a
graph.  When we talk about random regular graphs we mean a graph constructed
via the following process. Take $\Delta n$ vertices (assume $\Delta n$ is
even), numbered $0,\ldots,\Delta n-1$.  First, select a perfect matching on
these vertices uniformly at random among all perfect matchings.  Then, for all
$k\in \{0,\ldots,n\}$ merge the vertices $k\Delta, k\Delta +1, \ldots,
(k+1)\Delta-1$ into a single vertex of degree $\Delta$. Observe that though
this process may produce a multi-graph, the probability that it produces a
simple graph is bounded away from 0 as $n\to \infty$ and all simple
$\Delta$-regular graphs are equiprobable (\cite{bollobas1988isoperimetric}).
The probability space therefore consists of $(\Delta n -1)(\Delta n-3)\ldots$
matchings. We denote this product by $(\Delta n)!!$ (and generally we denote by
$n!!$ the product of all odd positive integers which are less than or equal to
$n$). We use $\log n$ to denote the binary logarithm of $n$ and $\ln n$ to
denote the natural logarithm of $n$.

Let $S\subseteq V$ be a set of vertices with $|S|\le n/2$. Let $c(S)$ denote
the number of edges with exactly one endpoint in $S$. The edge expansion of $S$
is defined as $\frac{c(S)}{|S|}$. The edge expansion of $G$ is defined as
$i(G)=\min_{|S|\le n/2} \frac{c(S)}{|S|}$.

Given a partition of $V$ into $S, V\setminus S$ we define the out-degree of a
vertex as the number of neighbors the vertex has on the other side of the
partition. We will characterize cuts mainly by the number of vertices of each
out-degree in $S$ and $V\setminus S$.  For this, we will use
$(\Delta+1)$-dimensional vectors, denoted as $\vec{s}=(s_0, s_1, \ldots,
s_\Delta)$. 

More specifically, consider a set $S\subseteq V$, with $|S|\le n/2$.  Let $s_i,
i\in \{0,\ldots,\Delta\}$ denote the number of vertices in $S$ that have
exactly $i$ neighbors in $V\setminus S$. We write
$\vec{s}=(s_0,\ldots,s_\Delta)$ and we have $ \sum_i s_i = |S|$.  Note that by
definition $\sum_i i s_i =  c(S)$.  Also, we denote by $s_i', i\in
\{0,\ldots,\Delta\}$ the number of vertices in $V\setminus S$ that have exactly
$i$ neighbors in $S$ and let $\vec{s'}=(s_0',\ldots,s_\Delta')$.  Note that
$\sum_i i s_i' = c(S)$ and $\sum_i s_i+\sum_i s_i'=n$.  We use $d$ to denote
the maximum $i$ such that $s_i>0$ and $d'$ to denote the maximum $i$ such that
$s_i'>0$.

\section{Main Analysis}

In this section we establish the main tools we will need in the analysis of the
expansion of random regular graphs. The main result gives a general form for
the distribution of out-degrees $\vec{s}$ that has maximum probability,
conditioned on the size of the cut (Lemma \ref{lem:localmax}). In the end of
this section we apply this tool to re-prove the main theorem of
\cite{bollobas1988isoperimetric}, giving the additional insight that the
distribution of out-degrees in the configuration of maximum probability is a
binomial distribution. In the following sections we use Lemma
\ref{lem:localmax} together with some observations on $d,d'$ to obtain improved
results. However, for now we will work only with the basic case where there is
no constraint on the out-degress and $d=d'=\Delta$.  As a result all sums given
in this section can be read to range from $i=0$ to $\Delta$.

We are given a random $\Delta$-regular graph $G(V,E)$. We want to examine the
asymptotic (in $n$) behavior of the expansion of such $\Delta$-regular graphs.
In other words, for each fixed $\Delta$ we would like to have a lower bound on
a random graph's expansion which holds with high probability for sufficiently
large graphs.  For concreteness, say we always assume that $n
> \Delta^{20}$. To avoid inessential difficulties let us also assume that $n$
is even.


Fix a set of vertices $S$. Our first step is to write down exactly the
probability that a certain configuration of out-degrees occurs for $S$ and
$V\setminus S$ and that it produces a cut of size $c(S) = c$.  More precisely,
we will calculate the probability that, given $\vec{s}=(s_0, s_1,
\ldots,s_\Delta)$ and $\vec{s'}=(s_0, s_1, \ldots,s_\Delta)$ such that $\sum_i
s_i = |S|$, $\sum_i s_i' = |V\setminus S|$ and $\sum_i is_i = \sum_i is_i' = c$
the set $S$ has exactly $c$ edges coming out, and these are distributed in a
such a way that for all $i$ there are $s_i$ vertices in $S$ with $i$ neighbors
in $V\setminus S$ and $s_i'$ vertices in $V\setminus S$ with $i$ neighbors in
$S$.  That probability is 

$$ P(\vec{s}, \vec{s'}) =  \frac{|S|!}{s_0!s_1!\ldots
s_\Delta!}\prod_{i=0}^{\Delta} {\Delta \choose i} ^{s_i} \frac{|V\setminus
S|!}{s_0'!s_1'!\ldots s_\Delta'!}\prod_{i=0}^{\Delta} {\Delta \choose i}
^{s_i'} c! \frac{(\Delta |S|-c)!!(\Delta |V\setminus S|-c)!!}{(\Delta n)!!}$$

Let us explain this. The first factor comes from deciding how to partition the
set $S$ into groups of vertices with the prescribed number of neighbors in
$V\setminus S$. Recall that each vertex in our graph is derived from a group of
$\Delta$ vertices in the original configuration where we take a random perfect
matching. Therefore, if we decide that a vertex has $i$ edges to vertices in
$V\setminus S$ we must decide which of the $\Delta$ vertices of the group are
the endpoints of these edges. This gives the second factor. The third and
fourth factors are similar for $V\setminus S$. The $c!$ comes from deciding how
to match the vertices from each side that have neighbors on the other.
Finally, we take the product of the total number of matchings for the rest of
the vertices of $S$ and $V\setminus S$. This is divided by the total number of
matchings overall. Recall now that we have written $P(\vec{s},\vec{s'})$ for a
specific set of vertices $S$. By union bound, the probability that some "bad"
set of vertices exists, such that its size is $u$, the cut has $c$ edges and
the out-degree distributions of the set and the complement are described by
$\vec{s}, \vec{s'}$ is at most ${n\choose u} P(\vec{s},\vec{s'})$.

Let $\mathbb{U}(u,c)$ be the set of all pairs $(\vec{s}, \vec{s'})$ such that
$\sum_i s_i=u$, $\sum_i (s_i+s_i')=n$ and $\sum_i is_i = \sum_i is_i' = c$.  In
other words, $\mathbb{U}$ is the set of all possible vectors describing a
configuration where $|S|=u$ and there are $c$ edges crossing the cut. Suppose
we could show that 

\begin{eqnarray}
 \lim_{n\to \infty} \sum_{u=0}^{n/2} \ \ \sum_{c=0}^{\lfloor \phi u \rfloor} \
\ \sum_{(\vec{s},\vec{s'})\in \mathbb{U}(u,c)} {n\choose u} P(\vec{s},\vec{s'})
&=& 0 \label{eqn:sum}
\end{eqnarray}

for some parameter $\phi$ depending on $\Delta$. Then we would have that with
high probability a random $\Delta$-regular graph has expansion at least $\phi$
from a simple application of the union bound. Our goal is to determine the
largest possible parameter $\phi(\Delta)$ such that this holds, and recall that
by the results of \cite{bollobas1988isoperimetric} and \cite{alon1997edge} we
expect the correct answer to be $\frac{\Delta}{2}-\Theta(\sqrt{\Delta})$.

Observe that the above sum has at most $\Delta n^{2\Delta+2}$ terms, since
$u\le n$, $c\le \Delta u$ and there are at most $n^\Delta$ different vectors
$\vec{s}$ such that $|s|= u$ for all $u$. Let $\vec{s_M}, \vec{s_M'}$ be such
that $P(\vec{s_M},\vec{s_M'})$ is maximum among all $(\vec{s},\vec{s'}) \in
\mathbb{U}(u,c)$ for all $u\le n/2$ and $c\le \phi u$. It is sufficient for us
to have that ${n\choose \frac{n}{2}}P(\vec{s_M},\vec{s_M'}) = o( \Delta^{-1}
n^{-2\Delta-2}) $.  In fact, we will have the much stronger condition
${n\choose \frac{n}{2}}P(\vec{s_M},\vec{s_M'}) \le \alpha^{-\Delta n}$ for some
$\alpha>1$, or in other words we will show that the probability that a bad set
exists is exponentially small.

\medskip

We have therefore reduced the question to that of determining a bad
configuration with maximum probability and upper-bounding that probability. We
would like to look for the vectors $\vec{s},\vec{s'}$, while avoiding some
extreme cases. We can do that by using the fact that small perturbations in
these vectors do not affect the probability too much. This is shown in the
following lemma which essentially shows that changing some of the coordinates
of a configuration vector $(\vec{s},\vec{s'})$ has a small impact on the
probability.

\begin{lemma} \label{lem:small}

If the configuration vectors $(\vec{s}, \vec{s'}), (\vec{p}, \vec{p'})$ satisfy
$\max_i \{|s_i-p_i|\} \le n^\frac{3}{4}$ and $\max_i \{|s_i'-p_i'|\} \le
n^\frac{3}{4}$ we have $P(\vec{s},\vec{s'}) \le 2^{o(  n)}
P(\vec{p},\vec{p'})$.

\end{lemma}

Lemma \ref{lem:small} allows us to neglect some annoying cases, such as for
example sets with too small size ($o(n^\frac{3}{4})$). We now only need to find
the configuration of maximum probability among configurations where
$s_i>n^\frac{3}{4}$ for all $i$ and make sure that the probability is
$2^{-\Theta(\Delta n)}$. Then, by Lemma \ref{lem:small} we will know that
looking for the maximum through the space of all configuration vectors would
have made no difference.

Now, given fixed values of $u,c$ let us try to find the configuration vector
$\vec{s}$ that gives maximum probability. It is not hard to see that we are
trying to maximize the following function:

$$ F(\vec{s}) = \frac{\prod_{i=0}^{\Delta} {\Delta \choose i}
^{s_i}}{s_0!s_1!\ldots s_\Delta!}$$

In other words, we want to determine the distribution of the out-degrees of $u$
vertices so that $\sum i s_i = c$ and $F(\vec{s})$ is maximum.  Note that the
problem for the vector $\vec{s'}$ is essentially identical.

\begin{lemma} \label{lem:localmax}

Let $\vec{s^M}$ be a configuration vector such that $\sum_i s^M_i = u$, $\sum_i
is^M_i = c$, for all $i, s^M_i>n^\frac{3}{4}$ and $F(\vec{s^M})$ is maximum
among all vectors that satisfy the previous conditions. Then $|s^M_i - s^M_0
\left( \frac{s^M_1}{\Delta s^M_0} \right)^i {\Delta \choose i}| =
o(n^\frac{3}{4})$ for all $i, 0\le i \le \Delta$.

\end{lemma}

\begin{proof}

We define $\vec{p^i}$ for $i\ge 2$ to be the vector with $p^i_0=p^i_i=1$,
$p^i_1=p^i_{i-1}=-1$ and all other coordinates equal to 0. Intuitively,
$\vec{p^i}$ is a small perturbation vector, which does not change either the
size of $S$ or the size of the cut. Since $\vec{s^M}$ maximizes the function
$F$, it must be the case that both $\vec{s}^M+\vec{p}^i$ and
$\vec{s}^M-\vec{p}^i$ give smaller values. So we should have
$\frac{F(\vec{s^M})}{F(\vec{s^M+\vec{p}^i})}\ge 1$ and also
$\frac{F(\vec{s^M})}{F(\vec{s^M-\vec{p}^i})}\ge 1$.

The first inequality gives $\frac{(s^M_0+1)(s^M_i+1)}{s^M_1s^M_{i-1}}
\frac{\Delta {\Delta \choose i-1}}{{\Delta \choose i}}\ge 1$. The second
inequality gives $\frac{(s^M_1+1)(s^M_{i-1}+1)}{s^M_0s^M_{i}} \frac{{\Delta
\choose i}}{\Delta {\Delta \choose i-1}}\ge 1$.

Observe that $\frac{\Delta {\Delta \choose i-1}}{{\Delta \choose i}} =
O(\Delta^2)$. Also, $\frac{s^M_i}{s^M_js^M_k} = O(\frac{1}{\sqrt{n}})$ because
all coordinates of $s^M$ are at least $n^\frac{3}{4}$. Thus, the first
inequality gives $\frac{s^M_0s^M_i}{s^M_1s^M_{i-1}}\frac{\Delta {\Delta \choose
i-1}}{{\Delta \choose i}} \ge 1 - O(\frac{\Delta^2}{\sqrt{n}})$. The second
gives $\frac{s^M_1s^M_{i-1}}{s^M_0s^M_i}\frac{{\Delta \choose i}}{\Delta
{\Delta \choose i-1}} \ge 1 - O(\frac{\Delta^2}{\sqrt{n}})$. Together these
imply that $|\frac{s^M_0s^M_i}{s^M_1s^M_{i-1}}\frac{\Delta {\Delta \choose
i-1}}{{\Delta \choose i}}-1| = O(\frac{\Delta^2}{\sqrt{n}})$. Since the
right-hand side is very small, intuitively what we will do is to assume it is
essentially 0 and solve the recurrence relation.

More precisely, let $\vec{p}$ be the vector defined as follows: $p_0 = s^M_0,
p_1=s^M_1$ and $p_i = \frac{p_1p_{i-1}}{p_0}\frac{{\Delta \choose i}}{\Delta
{\Delta \choose i-1}}$ for all $i\ge 2$. It is then not hard to see (or verify
by induction) that $p_i = s^M_0 \left( \frac{s^M_1}{\Delta s^M_0} \right)^i
{\Delta \choose i}$.

We have that

$$ \left(1-O\left(\frac{\Delta^2}{\sqrt{n}}\right) \right)
\frac{s^M_1s^M_{i-1}}{s^M_0}\frac{{\Delta \choose i}}{\Delta {\Delta \choose
i-1}} \le s^M_i \le \left(1+O\left(\frac{\Delta^2}{\sqrt{n}}\right) \right)
\frac{s^M_1s^M_{i-1}}{s^M_0}\frac{{\Delta \choose i}}{\Delta {\Delta \choose
i-1}} $$

From this we get that $(1- O\left(\frac{\Delta^2}{\sqrt{n}}\right)) p_2 \le
s^M_2 \le (1+ O\left(\frac{\Delta^2}{\sqrt{n}}\right)) p_2$. By induction we
then have $(1- O\left(\frac{\Delta^2}{\sqrt{n}}\right))^{i} p_i \le s^M_i \le
(1+ O\left(\frac{\Delta^2}{\sqrt{n}}\right))^{i} p_i$. Using the fact that
$\Delta \le n^\frac{1}{20}$ we have $((1 \pm
O\left(\frac{\Delta^2}{\sqrt{n}}\right))^{\Delta} = 1\pm
O\left(\frac{\Delta^3}{\sqrt{n}}\right)$. Therefore, for all $i$ we have
$|s^M_i-p_i| \le p_i\cdot O\left(\frac{\Delta^3}{\sqrt{n}}\right) = O(\Delta^3
\sqrt{n}) = o(n^\frac{3}{4})$. \qed

 \end{proof}

Informally, as a result of Lemma \ref{lem:localmax} we can now assume that the
vector describing the configuration of maximum probability has $s_i = s_0
\left( \frac{s_1}{\Delta s_0} \right)^i {\Delta \choose i}$, since for
sufficiently large $n$ we know that the maximum has an edit distance from this
vector smaller than the one allowed by Lemma \ref{lem:small}. To ease notation
slightly we introduce the parameters $\beta, \gamma$ defined so that $s_0=\beta
\sum_i s_i$ and $s_1 = \gamma \Delta s_0$. 
We will also use $\beta', \gamma'$ similarly as parameters corresponding to the
vector $\vec{s}'$.

Going back to equation (\ref{eqn:sum}) we would like to show that the maximum
probability is achieved for bisections, that is, $u=n/2$. This is intuitively
unsurprising, and it is stated by the following lemma. 

\begin{lemma} \label{lem:bisections}

Let $\vec{b},\vec{b}'$ be two vectors such that $\sum_i b_i = \sum_i b_i' =
n/2$, $\sum_i ib_i = \sum_i ib_i' = c$ and these two vectors maximize the
quantity ${n\choose \frac{n}{2}}P(\vec{b},\vec{b}')$ among all vectors that
satisfy the previous two conditions. Then for all vectors $\vec{s},\vec{s}'$
with $\sum_i s_i = u \le n/2$, $\sum_i is_i = \sum_i is_i' = c$ we have
${n\choose u}P(\vec{s},\vec{s}') \le 2^{o( n)}{n\choose
\frac{n}{2}}P(\vec{b},\vec{b}')$.

\end{lemma}

Thanks to Lemma \ref{lem:bisections} we are only interested in partitions of
size exactly $n/2$, since these give the maximum probability of achieving a cut
of size $c$ (perhaps modulo a $2^{o( n)}$ factor which will prove
insignificant). Therefore, for any expansion factor $\phi$ the probability of
finding a set with expansion less than $\phi$ is maximized when that set has
size exactly $\frac{n}{2}$.

 It is not hard to see that if the partition of the vertices is a bisection
then the distribution of out-degrees that gives maximum probability is the same
for the two parts. In other words, if $\vec{s},\vec{s}'$ are the two vectors
that give maximum probability for a cut of size $c$ and $\sum_i s_i = \sum_i
s_i'$ then it must be that $\vec{s}=\vec{s}'$. Using this, the probability that
a bisection which cuts $c$ edges exists is upper bounded by

$$ P_B(\vec{s}) \le {n \choose n/2} \left(\frac{(n/2)!}{s_0!s_1!\ldots
s_\Delta!}\right)^2 \prod_{i=0}^\Delta {\Delta \choose i}^{2s_i}
\frac{c!(\frac{\Delta n}{2}-c)!}{(\Delta n)!!} $$

where we have used the fact that $(n!!)^2<n!$. Now we would like to find the
range of values of $c$ for which the above quantity is very small, that is
$2^{-\Theta(\Delta n)}$. To ease notation, we write $c=(1-\eta)\frac{\Delta
n}{4}$, where $\eta >0$, and look for the correct range of $\eta$.

Using the facts that for large $n$ we have $\left( \frac{n}{e} \right)^n
<n!<\left( \frac{n}{e} \right)^npoly(n)$, ${n\choose n/2}< 2^n$ and $n!!>
\left( \frac{n}{e} \right)^{\frac{n}{2}} poly(n)$ we get that 

\begin{eqnarray*}
\log P_B(\vec{s})& \le& n + n \log\frac{n}{2e} - 2\sum_{i=0}^\Delta
s_i\log\frac{s_i}{e{\Delta \choose i}} + \frac{(1-\eta)\Delta n}{4}\log
\frac{(1-\eta)\Delta n}{4e} \\
&& + \frac{(1+\eta)\Delta n}{4}\log \frac{(1+\eta)\Delta n}{4e} - \frac{\Delta n}{2}\log\frac{\Delta n}{e} + o(n)
\end{eqnarray*}

As mentioned, assuming that $s_i=\beta \gamma^i {\Delta \choose i}\frac{n}{2}$
cannot affect the probability by more than $2^{o( n)}$, so with some
calculations and using the fact that $\sum s_i=\frac{n}{2}$ and $\sum
is_i=(1-\eta)\frac{\Delta n}{4}$  the upper bound becomes


 \begin{eqnarray} \log P_B &\le& n -n\log \beta - (1-\eta)\frac{\Delta
n}{2}\log{\gamma } -\Delta n\nonumber\\ && +\frac{(1+\eta)\Delta
n}{4}\log(1+\eta)+\frac{(1-\eta)\Delta n}{4}\log(1-\eta)+ o( n)
\label{eqn:prob1} \end{eqnarray}

What remains now is to find the values of $\beta,\gamma$ that give the maximum
probability. Thankfully, because $d=d'=\Delta$ we can use the binomial theorem
to find a clean solution to this problem as follows. Recall that $\sum_i s_i =
\sum_i \beta \gamma^i {\Delta \choose i} \frac{n}{2} = \beta (1+\gamma)^\Delta
\frac{n}{2} = \frac{n}{2}$.  Also, $\sum_i is_i = \sum_i i \beta \gamma^i
{\Delta \choose i} \frac{n}{2} = \sum_i \Delta \beta \gamma^i {\Delta-1 \choose
i-1} \frac{n}{2}= \beta \gamma \Delta (1+\gamma)^{\Delta-1}\frac{n}{2} =
(1-\eta)\frac{\Delta n}{4}$. Dividing these two equations gives $\gamma =
\frac{1-\eta}{1+\eta}$ and then $\beta=\left(\frac{1+\eta}{2}\right)^\Delta$.
Plugging these values into the upper bound for $P_B$ we get

\begin{eqnarray*} 
\frac{4log P_B}{\Delta n} &\le& \frac{4}{\Delta} - (1-\eta)\log(1-\eta) - (1+\eta)\log(1+\eta) + o(1/\Delta)
\end{eqnarray*}

The probability is therefore $2^{-\Omega(n)}$ whenever the right hand side is
less than some arbitrarily small negative constant.  \begin{theorem}
\label{thm:repeat}

For each $\Delta$, if $\eta>0$ satisfies
$(1-\eta)\log(1-\eta)+(1+\eta)\log(1+\eta) > \frac{4}{\Delta}$ then, for all
sufficiently large $n$, random $n$-vertex $\Delta$-regular graphs have
expansion at least $(1-\eta)\frac{\Delta}{2}$ with high probability.

\end{theorem}

Thus, we have obtained again the main result of
\cite{bollobas1988isoperimetric}, albeit through a much more winding road. One
bonus is that through Lemma \ref{lem:localmax} we now have a better
understanding of the out-degree distribution of maximum probability.
Informally, one way to interpret this is to observe that by Lemma
\ref{lem:localmax} the fraction of vertices that have out-degree $i$ is
proportional to the probability that a binomial random variable with parameters
$(\Delta,\frac{\gamma}{\gamma+1})$ takes the value $i$. This intuition will
come in handy in the asymptotic analysis.

\section{Local Improvements}

In this section we introduce the idea of local improvements in the analysis. As
mentioned, the main idea is to bound the probability that we find a partition
that not only has small expansion, but also is locally optimal, in that
flipping a pair of vertices cannot make the expansion smaller. The consequence
of this property that we will rely on is that for a locally optimal set all
vertices have out-degress at most (roughly) $\frac{\Delta}{2}$. The analysis of
the previous section will come in handy now, since the only difference is that
$d$ and $d'$ are now smaller. Essentially all of it goes through unchanged, up
to the point where we used the binomial theorem to find the values of $\beta,
\gamma$. Unfortunately, finding a clean formula for the values of
$\beta,\gamma$ in our case now becomes a much harder task. In this section we
will be dealing with specific values of $\Delta$, and we will therefore be able
to calculate $\beta,\gamma$ and the minimum $\eta$ for which the graph is whp
an $(1-\eta)\frac{\Delta}{2}$ expander numerically. In the next section we will
discuss ways to bound the values of $\beta,\gamma$ to obtain an improved
asymptotic bound as $\Delta$ tends to infinity.

We will say that a set $S\subseteq V$, $|S|\le n/2$ is locally improvable if
there exist $u\in S$, $v\in V\setminus S$ such that $S\setminus\{u\}\cup\{v\}$
has a strictly smaller cut. If $S$ is not locally improvable we will say that
it is locally optimal. Recall that $d$ denotes the maximum $i$ such that
$s_i>0$ and $d'$ the maximum $i$ such that $s_i'>0$. Observe that whenever
there exists a set with expansion $\phi$ there must also exist a locally
optimal set with the same size and expansion $\le \phi$.

\begin{lemma} \label{lem:improve}

If $S$ is locally optimal then $d+d'\le \Delta + 1$. Therefore, if $S$ is
locally optimal then $\min\{d,d'\} \le \lceil \frac{\Delta}{2} \rceil$.

\end{lemma}

\begin{proof}

Suppose for contradiction that $d+d'\ge \Delta + 2$. Let $u\in S$ have $d$
neighbors in $V\setminus S$ and $v\in V\setminus S$ have $d'$ neighbors in $S$.
We swap $u$ and $v$.  First, suppose that $u,v$ are not connected. The size of
the new cut is $c - d - d' + \Delta - d + \Delta - d' = c + 2\Delta - 2(d+d') <
c$ contradicting the local optimality of $S$. If $u,v$ are connected then the
new cut has size $c - (d-1) - (d'-1)   + \Delta - d + \Delta - d' < c$, again
contradicting the local optimality of $S$. \qed

\end{proof}

In fact, we can say slightly more than what Lemma \ref{lem:improve} states.
Notice that the thorny case is when the two vertices of maximum out-degrees
happen to be connected. If there are $u\in S, v\in V\setminus S$ with
out-degrees $d, d'$ respectively such that $u,v$ are not connected then with
the same argument we can say that if $S$ is locally optimal then $d+d'\le
\Delta$. Now, if there are at least $\Delta+1$ different vertices in $S$ with
out-degree $d$, then we can always find such a pair $u,v$ and conclude that for
a locally optimal set $d+d'\le \Delta$. Furthermore, if there are at most
$\Delta$ vertices of out-degree $d$ in $S$, then since $\Delta\le
n^{\frac{1}{20}} = o(n^\frac{3}{4})$ we can use Lemma \ref{lem:small} to say
that the probability is essentially unchanged if we assume that there are no
vertices of out-degree $d$, therefore the maximum out-degree in $S$ becomes
$d-1$. As a result, we will from now on assume that $d+d'\le \Delta$ for
locally optimal sets.

As mentioned, the analysis of the previous section goes through: even if
vertices in $S$ have maximum out-degree $d$ and vertices in $V\setminus S$ have
maximum out-degree $d'$ the probability of a certain configuration is still
robust to small perturbations (Lemma \ref{lem:small}) and the most likely
degree distribution is still in of the form $s_i = \beta \gamma^i {\Delta
\choose i} |S|$ (Lemma \ref{lem:localmax}). Also, bisections are again the
interesting case, so we may assume that $|S|=\frac{n}{2}$.

Revisiting the calculation that brought us to inequality (\ref{eqn:prob1}) we
now get:

\begin{eqnarray}
 \frac{\log P_B}{n} &\le& 1 -\frac{1}{2}\log \beta -\frac{1}{2}\log \beta' - (1-\eta)\frac{\Delta }{4}\log{\gamma }- (1-\eta)\frac{\Delta }{4}\log{\gamma' }  -
\Delta \nonumber\\
&& +\frac{(1+\eta)\Delta
}{4}\log(1+\eta)+\frac{(1-\eta)\Delta }{4}\log(1-\eta)+ o(1) \label{eqn:prob2}
\end{eqnarray}

Furthermore, we have

\begin{eqnarray}
\sum_{i=0}^d \beta \gamma^i {\Delta \choose i} = \sum_{i=0}^{d'} \beta' \gamma'^i {\Delta \choose i} &&= 1 \label{eqn:beta1}\\
\sum_{i=1}^d \beta i \gamma^i {\Delta \choose i} = \sum_{i=1}^{d'} \beta' i \gamma'^i {\Delta \choose i} &&= \frac{1-\eta}{2}\Delta \label{eqn:beta2}
\end{eqnarray}

Recall that in the previous section where we had $d=d'=\Delta$ we were able to
use the binomial theorem to simplify the above equations and solve for
$\beta,\gamma,\beta',\gamma'$ as functions of $\eta, \Delta$. Plugging the
result back into (\ref{eqn:prob1}) and asking that the right hand side is
negative led to Theorem \ref{thm:repeat}. Here, it is not clear how to do
something similar. 

However, if $\Delta,d,d'$ have some known fixed values, we can do the
following: pick a candidate value for $\eta$ and solve (numerically) equations
(\ref{eqn:beta1},\ref{eqn:beta2}) for $\beta,\gamma,\beta',\gamma'$. Plug the
solutions into (\ref{eqn:prob2}) and check if the right hand side is negative.
In such a case we can conclude that whp a random $\Delta$-regular graph will
not have a bisection with maximum out-degrees $d,d'$ and $\frac{(1-\eta)\Delta
n}{4}$ edges crossing the cut. Given that in a locally optimal set $d+d'\le
\Delta$ and allowing the maximum out-degree to increase can only increase the
probability, we are essentially looking for the minimum $\eta$ such that the
right hand side of (\ref{eqn:prob2}) is negative for all possible pairs $d,d'$
with $d+d'=\Delta$. 

Performing this procedure for various small values of $\Delta$ we get the
results listed in Table \ref{tbl:numbers}. Despite the fact that (as with the
bound given in \cite{bollobas1988isoperimetric}) we need to rely on some
numerical polynomial equation solver to find these numbers, their correctness
can be verified relatively easily. Specifically, these bounds can be verified
by checking that $\sum_{i=0}^d\beta \gamma^i {\Delta \choose i} = 1$,
$\sum_{i=1}^d\beta i \gamma^i {\Delta \choose i} = \frac{1-\eta}{2}\Delta$, the
same conditions apply for $\beta',\gamma'$ and the logarithm of the probability
(as given in (\ref{eqn:prob2})) evaluates to a negative number.

\begin{table}
\begin{tabular}{|lll|l|l|l|l|l|l|l|}
\hline
$\Delta$	& $d$	& $d'$	& $\eta$	& $\beta$	& $\gamma$	& $\beta'$	& $\gamma'$	& $i$	& Previous $i$ \\
\hline
4		& 	& 	& 0.778		&       	&       	&		&		& 0.444 	& 0.440 \\
		& 2	& 2	& 		& 0.61807	& 0.12938	&		&		&		& \\
		& 1	& 3	& 		& 0.55600	& 0.19964	& 0.62432	& 0.12503	&		& \\
\hline
5		& 	& 	& 0.701		&       	&       	&		&		& 0.7475	& 0.730 \\
		& 2	& 3	& 		& 0.41817	& 0.19904	& 0.44211	& 0.17786	&		& \\
		& 1	& 4	& 		& 0.25250	& 0.59208	& 0.44488	& 0.17587	&		& \\
\hline
6		& 	& 	& 0.648		&       	&       	&		&		& 1.056		& 1.041 \\
		& 3	& 3	& 		& 0.30304	& 0.22263	& 		&		&		& \\
		& 2	& 4	& 		& 0.25436	& 0.28521	& 0.31196	& 0.21445	&		& \\
\hline
7		& 	& 	& 0.599		&       	&       	&		&		& 1.4035	& 1.372 \\
		& 3	& 4	& 		& 0.18683	& 0.27917	& 0.20505	& 0.25492	&		& \\
		& 2	& 5	& 		& 0.11337	& 0.46594	& 0.20842	& 0.25117	&		& \\
\hline
8		& 	& 	& 0.565		&       	&       	&		&		& 1.740		& 1.716\\
		& 4	& 4	& 		& 0.13258	& 0.29015	& 		&		&		& \\
		& 3	& 5	& 		& 0.10607	& 0.34576	& 0.13928	& 0.27975	&		& \\
		& 2	& 6	& 		& 0.02515	& 1.0425	& 0.14044	& 0.27812	&		& \\
\hline
9		& 	& 	& 0.531		&       	&       	&		&		& 2.1105	& 2.0655 \\
		& 4	& 5	& 		& 0.07715	& 0.33700	& 0.08734	& 0.31226	&		& \\
		& 3	& 6	& 		& 0.04607	& 0.46801	& 0.08982	& 0.30718	&		& \\
\hline
10		& 	& 	& 0.507		&       	&       	&		&		& 2.465		& 2.430 \\
		& 5	& 5	& 		& 0.05415	& 0.34150	& 		&		&		& \\
		& 4	& 6	& 		& 0.04182	& 0.39127	& 0.05798	& 0.32989	&		& \\
		& 3	& 7	& 		& 0.01319	& 0.71814	& 0.05885	& 0.32751	&		& \\
\hline
\end{tabular} \hspace{0.5cm}
\begin{tabular}{|l|l|l|l|}
\hline
$\Delta$	&	$\eta$	&	$i$ 	& Previous $i$ \\
\hline
11				& 0.482								& 2.849		& 2.794 \\
12		 	 	& 0.464		       	       					& 3.216		& 3.168 \\
13		 	 	& 0.444		       	       					& 3.614		& 3.549 \\
14		 	 	& 0.430		       	       					& 3.99		& 3.934 \\
15		 	 	& 0.414		       	       					& 4.395		& 4.320 \\
16		 	 	& 0.403		       	       					& 4.776		& 4.712 \\
17		 	 	& 0.389		       	       					& 5.1935	& 5.1085 \\
18		 	 	& 0.380		       	       					& 5.580		& 5.508 \\
19		 	 	& 0.368		       	       					& 6.004		& 5.909 \\
20		 	 	& 0.360		       	       					& 6.4		& 6.320 \\
30		 	 	& 0.294		       	       					& 10.59		& 10.47 \\
40		 	 	& 0.255		       	       					& 14.9		& 14.74 \\
\hline

\end{tabular}

\caption{Summary of numerically obtained lower bounds for various values of
$\Delta$.  The last column contains the numbers that follow by numerically
solving the condition given in \cite{bollobas1988isoperimetric}.  Due to space
constraints we only list the final results for $\Delta\ge 12$.}
\label{tbl:numbers}

\end{table}

Observe that the improvements we obtain, though quite modest (in the order of
$10^{-2}$ for $\Delta \le 10$ and $10^{-1}$ for somewhat larger $\Delta$) seem
to grow with $\Delta$. This is a fact explained also by the results of the next
section. 

\section{Asymptotic Bound}

In this section we will extend the local improvement analysis to general
$\Delta$ and show an asymptotic (in $\Delta$) lower bound on the expansion of
random $\Delta$-regular graphs. As mentioned, this is interesting in part
because it confirms the findings of the previous section that the gains from
this analysis increase with $\Delta$. 

Beyond this however, determining a tight asymptotic bound on this quantity is a
fundamental mathematical question. Recall that the lower bound given by
Bollobás is $\frac{\Delta}{2}-\sqrt{\Delta}\sqrt{\ln 2}$. It is very natural to
ask where the $\sqrt{\ln2}$ comes from and whether it is the "right" constant
here.  One intuitive explanation might be the following: consider a bisection
of the set of vertices. Each edge of the random graph has probability
(essentially) $\frac{1}{2}$ of crossing the cut, therefore the expected number
of edges crossing the cut is $\mu = \frac{\Delta n}{4}$. If the edges were
independent then the size of the cut would follow a binomial distribution with
$\sigma^2=\frac{\Delta n}{8}$. What is the maximum size $c$ such that we can
guarantee that a cut of this size has probability $o(2^{-n})$ of occurring? If
we approximate the binomial distribution by a normal distribution we get that
$\frac{1}{\sqrt{2\pi}\sigma} e^{-\frac{1}{2}(\frac{\mu-c}{\sigma})^2}$ must be
$o(2^{-n})$. Solving this gives
$c=\frac{n}{2}(\frac{\Delta}{2}-\sqrt{\Delta}\sqrt{\ln 2})$. In other words,
the analysis of \cite{bollobas1988isoperimetric} shows that the true
probability that a fixed set of vertices has $c$ edges coming out can be
approximated very well by assuming that edges are independent, an interesting
and natural result.

Nevertheless, as mentioned the weakness in the analysis of
\cite{bollobas1988isoperimetric} is the union bound, which would only be tight
if the probability that two different sets of vertices are expanding were
independent. In some way the local improvement idea plays on this
non-independence and, though we make only very small progress towards bridging
the gap between the known lower and upper bound, our results establish in an
indirect way that $\sqrt{ln2}$ is not the right constant and that there is
something significant lost by the union bound in the analysis of
\cite{bollobas1988isoperimetric}.

To simplify presentation, in this section we define a locally optimal partition
as one where $\min\{d,d'\}\le\lceil\frac{\Delta}{2}\rceil$. This is a slightly
weaker definition than what we use in the previous section. Assume without loss
of generality that $d\le d'$. Following similar reasoning as before, if there
exists a non-expanding set there must exist one with
$d\le\lceil\frac{\Delta}{2}\rceil$, and since allowing $d,d'$ to be larger can
only increase the probability we can assume that $d=\frac{\Delta}{2}$ and
$d'=\Delta$. Informally, we are applying the local improvement argument only to
one side of the partition to simplify things. In particular, now we don't need
to check various combinations of $d,d'$, or to prove that the maximum
probability is achieved when $d=d'$ (which would agree with the numerical
results of the previous section but would likely be complicated to prove
directly). Finally, as usual assume that we are dealing with bisections, since
these are the most interesting case and we will assume that $\eta <
\frac{2\sqrt{\ln2}}{\sqrt{\Delta}}$, since otherwise we already know from
\cite{bollobas1988isoperimetric} that the probability of such a set tends to 0.

Having set $d'=\Delta$ we can again use the binomial theorem and equations
(\ref{eqn:beta1},\ref{eqn:beta2}) to get $\gamma'=\frac{1-\eta}{1+\eta}$ and
$\beta'=\left(\frac{1+\eta}{2}\right)^\Delta$. Inserting these into
(\ref{eqn:prob2}) gives

\begin{eqnarray}
 \frac{\log P_B}{n} &\le& 1 - \frac{\Delta}{2} - \frac{1}{2}\log\beta - (1-\eta)\frac{\Delta}{4}\log\gamma +  o(1) \label{eqn:prob3}
\end{eqnarray}

The question therefore now becomes to determine $\beta,\gamma$ and the smallest
possible $\eta$ so that the right hand side of (\ref{eqn:prob3}) is negative.
Going back to the left hand side of equations (\ref{eqn:beta1},\ref{eqn:beta2})
we observe that

$$\sum_{i=0}^d \beta \gamma^i {\Delta \choose i} = \beta (\gamma+1)^\Delta
\sum_{i=0}^d \left(\frac{\gamma}{\gamma+1}\right)^i
\left(\frac{1}{\gamma+1}\right)^{\Delta-i} {\Delta \choose i}$$

We also have that 

$$\sum_{i=1}^d \beta i \gamma^i {\Delta \choose i} = \beta \Delta \sum_{i=1}^d
\gamma^i {\Delta-1 \choose i-1} = \beta \gamma \Delta (\gamma+1)^{\Delta-1}
\sum_{i=0}^{d-1} \left(\frac{\gamma}{\gamma+1}\right)^i
\left(\frac{1}{\gamma+1}\right)^{\Delta-i-1} {\Delta-1 \choose i}$$

where we used the fact that ${\Delta \choose i}=\frac{\Delta}{i}{\Delta-1
\choose i-1}$.

The main intuition we are going to use now is that the two sums can be written
more cleanly as probabilities using the binomial distribution. Specifically,
let $P_1$ be defined as $P_1 = \mathbf{Pr}[B(\Delta,\frac{\gamma}{\gamma+1})\le
d]$, where by $B(n,p)$ we denote a random variable that follows the binomial
distribution with parameters $n,p$. Also, let $P_2$ be defined as $P_2 =
\mathbf{Pr}[B(\Delta-1,\frac{\gamma}{\gamma+1})\le d-1]$. Informally, as
$\Delta$ tends to infinity we expect $P_1$ and $P_2$ to become almost equal.

Now, with the help of the above equations (\ref{eqn:beta1},\ref{eqn:beta2})
become

\begin{eqnarray}
 \beta (\gamma+1)^\Delta  P_1 &=& 1 \label{eqn:beta3}\\
 \beta  \gamma (\gamma+1)^{\Delta-1} P_2 &=& \frac{1-\eta}{2} \label{eqn:beta4}
\end{eqnarray}

Again, we will divide these two to get an equation for $\gamma$. To do this we
would like to establish a relationship between $P_1$ and $P_2$. Let $P_3 =
\mathbf{Pr}[B(\Delta,\frac{\gamma}{\gamma+1})= d]$.

\begin{lemma} \label{lem:p1p3}

Let $P_1,P_2,P_3$ as defined above. Then $P_1=P_2+\frac{\Delta-d}{\Delta}P_3$.

\end{lemma}

To ease notation slightly let
$\theta=\frac{\Delta-d}{\Delta}\frac{P_3}{P_1}=\frac{P_3}{2P_1}$. Dividing
equations (\ref{eqn:beta3},\ref{eqn:beta4}) we have
$\frac{\gamma+1}{\gamma}\frac{1}{1-\theta}=\frac{2}{1-\eta}$. Solving this for
$\gamma$ gives $\gamma=\frac{1-\eta}{1+\eta-2\theta}$. From this we get
$\beta=\left(\frac{1+\eta-2\theta}{2-2\theta}\right)^\Delta \frac{1}{P_1}$.
Plugging these into (\ref{eqn:prob3}) we get

\begin{eqnarray*}
 \frac{\log P_B}{n} &\le& 1 + \frac{\Delta}{2}\log(1-\theta) + \frac{1}{2}\log P_1 - (1-\eta)\frac{\Delta}{4}\log(1-\eta) -  (1+\eta)\frac{\Delta}{4}\log(1+\eta-2\theta) +  o(1) 
\end{eqnarray*}

After re-writing all logarithms to base $e$ we get 

\begin{eqnarray*}
 \frac{4\ln P_B}{n\Delta} &\le& \frac{4\ln2}{\Delta} + 2\ln(1-\theta) + \frac{2}{\Delta}\ln P_1 - (1-\eta)\ln(1-\eta) -  (1+\eta)\ln(1+\eta-2\theta) +  o(\frac{1}{\Delta}) 
\end{eqnarray*}

With a little foresight assume that in the end we will set $\eta$ sufficiently
large to have $\eta>\theta + \frac{1}{2\sqrt{\Delta}}$. In this case we would
have $\frac{\gamma}{\gamma+1}=\frac{1}{2}\frac{1-\eta}{1-\theta} \le
\frac{1}{2} (1-\eta+\theta)$. Therefore,  $\frac{\Delta\gamma}{\gamma+1}\le
\frac{\Delta}{2} - \frac{\sqrt{\Delta}}{4}$. Thus, the expected value of
$B(\Delta,\frac{\gamma}{\gamma+1})$ differs by at least half a standard
deviation from $d$. We can use this to bound $\theta$ from above using the
normal distribution to approximate $P_1,P_3$ we have $\theta\le \frac{1}{2}
\sqrt{\frac{2}{\pi\Delta}}\frac{e^{-\frac{1}{8}}}{0.68}\le
\frac{0.52}{\sqrt{\Delta}}$.

Now, using the fact that both $\eta$ and $\theta$ are
$O(\sqrt{\frac{1}{\Delta}})$ we can use the expansion
$\ln(1-x)=-x-\frac{x^2}{2}-O(x^3)$ and the above becomes

\begin{eqnarray*}
 \frac{4\ln P_B}{n\Delta} &\le& \frac{4\ln2}{\Delta} + \frac{2}{\Delta}\ln P_1  + \theta^2 - \eta^2 +  o(\frac{1}{\Delta}) 
\end{eqnarray*}

So, we essentially need to have that $\eta^2 > \frac{4\ln2}{\Delta} +
\frac{2}{\Delta}\ln P_1 + \theta^2$. If the sum of the last two terms is
negative the bound for $\eta$ is strictly smaller that $\frac{2\sqrt{\ln
2}}{\sqrt{\Delta}}$. We can see that this is the case as follows. 

Recall that $\frac{\gamma}{\gamma+1}=\frac{1}{2}\frac{1-\eta}{1-\theta} \ge
\frac{1-\eta}{2}$. We are interested in the case $\eta<\frac{2\sqrt{\ln
2}}{\Delta}$. So, the expected value of $B(\Delta,\frac{\gamma}{\gamma+1})$
differs from $d=\frac{\Delta}{2}$ by at most $\sqrt{\ln2}\sqrt{\Delta}$ which
is at most 2 standard deviations. Again using a normal approximation for $P_1$
we get that $P_1\le 0.84 \Rightarrow \ln P_1 \le -0.17$. Using the upper bound
on $\theta$ we have given already we have that the sum of the two terms is at
most $-\frac{0.07}{\Delta}$.  Thus, the lower bound for $\eta$ is strictly
smaller than the one given in \cite{bollobas1988isoperimetric}.

\begin{theorem}

There exists an $\alpha < 2\sqrt{\ln 2}$ such that whp random $\Delta$-regular
graphs have expansion at least
$(1-\frac{\alpha}{\sqrt{\Delta}})\frac{\Delta}{2}$.

\end{theorem}

\section{Conclusions}

In this paper we gave improved lower bounds on the expansion of random
$\Delta$-regular graphs for small values of $\Delta$. We also showed that the
asymptotic lower bound given by Bollobás is not tight. Perhaps more
importantly, we gave an alternative analysis of the result of
\cite{bollobas1988isoperimetric} incorporating the distribution of out-degrees
and showing a general form. We believe that Lemma \ref{lem:localmax} could
prove to be useful in establishing other properties of random regular graphs as
well.

As to the problems of this paper, much remains to be done. First, we would like
to further refine our analysis, especially the asymptotic part. In particular,
it would be nice to prove that when the local improvement argument is applied
to both sides of the partition the maximum probability is achieved when $d=d'$.
This would agree with our numerical findings for small $\Delta$, but a general
proof has so far been elusive. Furthermore, it would be interesting to see if
the local improvement idea can be taken any further, especially for special
cases of small values of $\Delta$ where there may be room for more
combinatorial arguments. Finally, for a concrete open problem, it would be nice
to determine if there exist $5$-regular graphs with expansion at least 1. This
seems unlikely for random graphs, since the bound we give is still only
(almost) $0.75$ while there exists an upper bound of only slightly more than 1.
However, if such graphs exist they would prove very useful in the field of
inapproximability of bounded occurrence CSPs.

\textbf{Acknowledgement:} I am grateful to Johan Håstad who first pointed out to
me the idea of looking for a locally optimal set of vertices, which is the
basis of this paper.

\newpage

\bibliographystyle{abbrv} \bibliography{expanders}

\newpage

\appendix

\section{Omitted Proofs}

\subsection{Proof of Lemma \ref{lem:small}}

\begin{proof}

We want to upper bound $\frac{P(\vec{s},\vec{s'})}{P(\vec{p},\vec{p'})}$. Since
$|s_i-p_i| \le n^\frac{3}{4}$ we have $\frac{p_i!}{s_i!} \le
2^{O(n^\frac{3}{4}\log n)} $. There are $2\Delta$ factors of this form, so
their total contribution is $2^{O(\Delta n^\frac{3}{4}\log n)}= 2^{o( n)}$.
There are also $2\Delta$ factors of the form ${\Delta \choose i}^{s_i-p_i}$,
each of which contributes at most $2^{\Delta n^\frac{3}{4}}$. Thus, their total
contribution is at most $2^{\Delta^2 n^\frac{3}{4}} = 2^{o( n)}$, using that
$\Delta \le n^\frac{1}{20}$.  Similarly, the change in the remaining factors is
also upper-bounded by $2^{\Delta^2 n^\frac{3}{4}} = 2^{o( n)}$, since the
number of edges crossing the cut cannot change by more than $\Delta^2
n^\frac{3}{4}$. \qed

\end{proof}

\subsection{Proof of Lemma \ref{lem:bisections}}

\begin{proof}

Consider a pair of vectors $\vec{s}, \vec{s}'$ as described above. We can
assume without loss of generality that $s_i = \beta \gamma^i {\Delta \choose
i}u$ and similarly for $s_i'=\beta'\gamma'^i{\Delta \choose i}(n-u)$, since by
Lemma \ref{lem:localmax} and Lemma \ref{lem:small} the actual maximum cannot
differ from the maximum achieved by such vectors by a factor of more than
$2^{o( n)}$. 

First, assume that $\sum_i s_i = u < n/2 < \sum_i s_i'$. Therefore, we have

\begin{eqnarray*}
\sum_i \beta \gamma^i {\Delta \choose i}u &<& \sum_i
\beta'\gamma'^i {\Delta \choose i}(n-u) \\
\sum_i i \beta \gamma^i {\Delta \choose i}u &=& \sum_i i
\beta'\gamma'^i {\Delta \choose i}(n-u) \\
\end{eqnarray*}

From these two we get that $\beta u< \beta'(n-u)$ or $s_0 < s_0'$. Consider now
the configuration that we can obtain from $\vec{s},\vec{s}'$ by decreasing
$s_0'$ by one and increasing $s_0$ by one, thus increasing $u$. We claim that
this new configuration achieves a value of ${n\choose u}P(\vec{s},\vec{s}')$
that is at least as high.  Indeed, because ${n\choose u}$ simplifies with the
factors $|S|!=u!$ and $|V\setminus S|!=(n-u)!$ into $n!$, the only factors
affected by this change are $\frac{1}{s_0!s_0'!}$. But since $s_0<s_0'$ this
product does not decrease by the change and we have a configuration with larger
$u$. Repeating this argument leads to $u=n/2$. \qed

\end{proof}

\subsection{Proof of Lemma \ref{lem:p1p3}}

\begin{proof}

From the definitions of $P_1,P_2$ it is not hard to see that
$P_1=P_2+\mathbf{Pr}[B(\Delta-1,\frac{\gamma}{\gamma+1})=
d-1]\cdot\frac{1}{\gamma+1}$. Intuitively, we will get at most $d$ successes in
$\Delta$ trials iff we either get at most $d-1$ in the first $\Delta-1$, or we
get exactly $d$ in the first $\Delta-1$ and fail in the last.

But now
$\mathbf{Pr}[B(\Delta-1,\frac{\gamma}{\gamma+1})=d-1]\cdot\frac{1}{\gamma+1}={\Delta-1\choose
d}\frac{\gamma^d}{(\gamma+1)^\Delta}=\frac{\Delta-d}{\Delta}{\Delta\choose
d}\frac{\gamma^d}{(\gamma+1)^\Delta}$, where we have used the fact that
${\Delta-1\choose d}=\frac{\Delta-d}{\Delta}{\Delta\choose d}$. \qed

\end{proof}

\end{document}